\documentclass{article}

\usepackage{amssymb}
\usepackage{amsmath}
\usepackage{amsthm}
\usepackage{algorithm}
\usepackage{algorithmic}
\usepackage{fullpage}

\newtheorem{thm}{Theorem}[section]
\newtheorem{clm}[thm]{Claim}

\newtheorem{lem}[thm]{Lemma}

\newtheorem{defn}[thm]{Definition}

\title{Exponential Inapproximability of Selecting a Maximum Volume
Sub-matrix}

\author{Ali {\c{C}}ivril\\
Meliksah University, Computer Engineering Department,\\
Talas, Kayseri 38280 Turkey \and
Malik Magdon-Ismail\\
Rensselaer Polytechnic Institute, Computer Science Department,\\
110 8th Street Troy, NY 12180 USA}

\begin{document}
\maketitle

\begin{abstract}
Given a matrix $A \in \mathbb{R}^{m \times n}$ ($n$ vectors in $m$
dimensions), and a positive integer $k < n$, we consider the
problem of selecting $k$ column vectors from $A$ such that the
volume of the parallelepiped they define is maximum over all
possible choices. We prove that there exists $\delta<1$ and $c>0$
such that this problem is not approximable within $2^{-ck}$ for $k
= \delta n$, unless $P=NP$.
\end{abstract}

%\begin{keyword}
%maximum volume sub-matrix, complexity, inapproximability
%\end{keyword}

%\end{frontmatter}
%%%%%%%%%%%%%%%%%%%%%%%%%%%%%%%%%%%%%%%%%%%%%%%%%%%%%%%%%%%%%%%%%%%%%%%%%%%%%%%%%%%%%%
%%%%%%%%%%%%%%%%%%%%%%%%%%%%%%%%%%%%%%%%%%%%%%%%%%%%%%%%%%%%%%%%%%%%%%%%%%%%%%%%%%%%%%
%%%%%%%%%%%%%%%%%%%%%%%%%%%%%%%%%%%%%%%%%%%%%%%%%%%%%%%%%%%%%%%%%%%%%%%%%%%%%%%%%%%%%%
\section{Introduction}
%Most data can be represented as an $m \times n$ matrix where the
%columns are objects and the rows are the features associated with
%them. Among the important examples of such representation in
%modern statistical analysis are document-term data, DNA microarray
%data and user-movie data where the analysts often need to define a
%feature vector for a specific object. Hence,
Given a matrix $A \in \mathbb{R}^{m \times n}$, it is of practical
importance to obtain the ``significant information'' contained in
$A$. It becomes especially important to have a compact
representation of $A$ when $A$ is large and has low numerical
rank, as is typical of modern data. Thus, in a broad sense, we are
interested in concise representations of matrices. Besides the
tremendous practical impact of linear algebraic algorithms
designed to this aim, they also come up in different theoretical
forms and paradigms. Specifically, the formalization of
``significant information'' can be done in several ways and to a
great extent, it depends on how a matrix is interpreted.

From a conceptual point of view, rather than interpreting a matrix
as a block of numbers, we view it as a set of vectors
(specifically, column vectors) which are indivisible entities.
Thus, the formalization of ``significant information'' is
essentially related to finding a subset of columns of the matrix
which satisfies some certain spectral conditions or orthogonality
requirements. From a purely combinatorial perspective, treating
vectors as elements of a set, one can also view subset selection
in matrices as a generalization of the usual subset selection
problem where the elements contain little or no information. To
give a specific example, the well known Set Cover problem asks for
a smallest cardinality subset of a set system which covers a
universal set. Likewise, the problem we are interested in
essentially asks for a small number of column vectors to ``cover''
the whole matrix. In this paper, we state a measure of quality for
this problem, namely the volume, and we prove an exponential
inapproximability result for the problem of selecting a maximum
volume sub-matrix of a matrix.

Several problems in matrix analysis require to construct a more
concise version of a matrix generally performed by a re-ordering
of the columns \cite{Golub}, such that the new smaller matrix is
as good a representative of the original as possible. One of the
criteria that defines the quality of a subset of columns of a
matrix is how well-conditioned the sub-matrix that they define is.
To motivate the discussion, consider the set of three vectors
$$ \textstyle \left\{ e_1= \left(
\begin{matrix}
1\\
0
\end{matrix}
\right), e_2= \left(
\begin{matrix}
0\\
1
\end{matrix}
\right), u= \left(
\begin{matrix}
\sqrt{1-\epsilon^2}\\
\epsilon
\end{matrix}
\right) \right\}, $$

\noindent which are clearly dependent, and any two of which are a
basis. Thus any pair can serve to reconstruct all vectors. Suppose
we choose $e_1, u$ as the basis, then
$e_2=(1/\epsilon)u-(\sqrt{1-\epsilon^2}/{\epsilon})e_1$, and we
have a numerical instability in this representation as
$\epsilon\rightarrow 0$. Such problems get more severe as the
dimensionality of the space gets large (curse of dimensionality),
and it is natural to ask the representatives to be ``as far away
from each other as possible''. From this simple example, we see
that two orthogonal vectors will capture more information about a
superset of columns than two that have an acute angle between each
other. Hence, in its generality, this vaguely stated problem can
be stated as finding a subset of columns with the maximum volume
possible or equivalently with the maximum determinant. A similar
(but not equivalent) problem is to find a subset with the maximum
smallest singular value. Indeed, in one of the early works
studying Rank Revealing QR (RRQR) factorizations \cite{Hong},
while discussing different options on how to choose a good
sub-matrix, it was noted that it turns out that ``the selection of
the sub-matrix with the maximum smallest singular value suggested
in \cite{Klema} can be replaced by the selection of a sub-matrix
with maximum determinant'', which heuristically proposes to
maximize the volume of the sub-matrix instead of using more
complicated functions. Several algorithms have been designed
following this intuition
\cite{Chan,Chan-Hansen,Chandrasekaran,Hoog,Gu,Hong,Pan-Tang}. The
optimization problem of finding a maximum volume sub-matrix of a
matrix was only recently studied by {\c{C}}ivril and Magdon-Ismail:

\begin{defn}\cite{Volume-TCS}
Given a matrix $A \in {\mathbb{R}}^{m \times n}$ of rank at least
$k$, MAX-VOL is the problem of finding a sub-matrix $C \in
\mathbb{R}^{m \times k}$ of $A$ such that the volume of the $k$
dimensional parallelepiped defined by the column vectors in $C$ is
maximum over all possible choices.
\end{defn}

\begin{thm}\cite{Volume-TCS}
MAX-VOL is NP-hard. Further, it is NP-hard to approximate to
within $2\sqrt{2}/3+\epsilon$ for arbitrarily small $\epsilon >
0$.
\end{thm}

Since MAX-VOL is NP-hard, it is natural to ask for an algorithm to
approximate the maximum volume. The first thing one might try is a
simple greedy algorithm for approximating MAX-VOL:

\begin{algorithm}[htb]
\caption{Greedy}
\begin{algorithmic}[1]
\STATE $C \leftarrow \emptyset$ \WHILE{$|C| < k$} \STATE Select
the largest norm vector $v \in A$ \STATE Remove the projection of
$v$ from every element of $A$ \STATE $C \leftarrow C \cup v$
\ENDWHILE
\end{algorithmic}
\end{algorithm}

The analysis of the approximation ratio of this algorithm and a
lower bound was also provided in \cite{Volume-TCS}. Specifically,
let $Vol(Gr)$ be the volume of the column vectors chosen by Greedy
and let $Vol(Opt)$ be the optimum volume. Then, we have

\begin{thm}\cite{Volume-TCS}
$Vol(Gr) \geq \frac{1}{k!} \cdot Vol(Opt)$.
\end{thm}

\begin{thm}\cite{Volume-TCS}
There exists an instance of MAX-VOL for which $Vol(Gr) \leq
\frac{1}{2^{k-1}} (1-\epsilon) \cdot Vol(Opt)$ for arbitrarily
small $\epsilon >0$. Furthermore, this instance can explicitly be
constructed.
\end{thm}

%%%%%%%%%%%%%%%%%%%%%%%%%%%%%%%%%%%%%%%%%%%%%%%%%%%%%%%%%%%%%%%%%%%%%%%%%%%%%
Note that there is a gap between the proven approximation ratio
and the lower bound implied by the explicit example. The analysis
yielding the ratio $1/k!$ is essentially a product of $k$
different mutually exclusive analyses related to each step of the
algorithm. However, it is not clear whether the overall
contribution of these different steps to the approximation ratio
is actually better than their products. Indeed, the lower bound of
$1/2^{k-1}$ pertains to such a peculiar construction that we have
conjectured a $1/2^{k-1}$ approximation ratio for the greedy
algorithm.  Hence, in general, proving an exponential
inapproximability for this problem is an important step towards
characterizing its approximability properties. It will show that
the greedy algorithm is almost the best one can hope for.

This work takes a first step towards this goal and prove
exponential inapproximability for MAX-VOL via a gap preserving
reduction from the well known Label-Cover problem using the
Parallel Repetition Theorem \cite{Raz}. In doing so, we will
establish that the greedy algorithm is asymptotically optimal up
to a logarithm in the exponent. Specifically, we prove the
following theorem:

\begin{thm}
There exists $\delta<1$ and $c>0$ such that the problem MAX-VOL is
not approximable within $2^{-ck}$ for $k = \delta n$, unless
$P=NP$.
\end{thm}

Our reduction may also be of independent interest which can be
used to prove inapproximability results for other matrix
approximation problems with different objective functions.

%%%%%%%%%%%%%%%%%%%%%%%%%%%%%%%%%%%%%%%%%%%%%%%%%%%%%%%%%%%%%%%%%%%%%%%%%%%%%
\subsection{Preliminaries and Notation}
We introduce some preliminary notation and definitions. Let a
matrix $A$ be given in column notation as: $A = \{v_1, v_2,
\hdots, v_n\}$. The volume of $A$, $Vol(A)$ can be recursively
defined as follows: if A contains one column, i.e. $A = \{v_1\}$,
then $Vol(A) = {\|v\|}_2$, where ${\| \cdot \|}_2$ is the
Euclidean norm. If A has more than one column, $ Vol(A) =
{\|v-\pi_{(A-\{v\})}(v)\|}_2\cdot Vol(A-\{v\})$ for any $v \in A$,
where $\pi_{A}(v)$ is the projection of $v$ onto the space spanned
by the column vectors of $A$. It is well known that
$\pi_{(A-\{v\})}(v)=A_v A_v^+v$, where $A_v$ is the matrix whose
columns are the vectors in $A-\{v\}$, and $A_v^+$ is the
pseudo-inverse of $A_v$ (see for example \cite{Golub}). Using this
recursive expression, we have

\begin{equation*} \label{volume-inapprox} Vol(S) = Vol(A) = {\|v_1\|}_2 \cdot \prod_{i=1}^{n-1}
{\|v_{i+1}-A_i A_i^+v_{i+1}\|}_2
\end{equation*}

\noindent where $A_i = \{v_1 \cdots v_i\}$ for $ \leq i \leq n-1$.

We observe a simple fact about the ``distance'' of a vector to a
subspace in the following lemma, which will be useful in the final
proof. Given two sets of vectors $P$ and $Q = \{q_1, \hdots ,
q_m\}$, let $d(q,P) = {\|q-\pi_P(q)\|}_2$ denote the distance of
$q \in Q$ to the space spanned by the vectors in $P$.

\begin{lem}[Union Lemma]
\label{union} $Vol(P \cup Q) \leq Vol(P) \cdot \prod_{i=1}^n
d(q_i, P)$.
\end{lem}

\begin{proof}
We argue by induction on $m$. For $m=1$, $Q$ has one element and
the statement trivially holds. Assume that it is true for $n=k$
where $Q = \{q_1, \hdots ,q_k\}$. Then, for any $q_{k+1}$

\begin{align*}
Vol(P \cup Q \cup \{q_{k+1}\}) &= Vol(P \cup Q) \cdot d(q_{k+1}, P \cup Q) \\
&\leq_{(a)} Vol(P \cup Q) \cdot d(q_{k+1}, P) \\
&\leq_{(b)} Vol(P) \cdot \prod_{i=1}^k d(q_i, P) \cdot d(q_{k+1},P) \\
&= Vol(P) \cdot \prod_{i=1}^{k+1} d(q_i,P).
\end{align*}

\noindent (a) follows because $d(q, A \cup B) \leq d(q, A)$ for
any $A$, $B$ and (b) follows by the induction hypothesis.
\end{proof}
%%%%%%%%%%%%%%%%%%%%%%%%%%%%%%%%%%%%%%%%%%%%%%%%%%%%%%%%%%%%%%%%%%%%%%%%%%%%%
\subsection{Related Work}
The concept of volume has been closely related to matrix
approximation and mainly studied from a linear algebraic
perspective. There are a few results revealing the relationship
between the volume of a subset of columns of a matrix and its
approximation. In \cite{Vempala}, the authors introduced
\emph{volume sampling} to find low-rank approximation to a matrix
where one picks a subset of columns with probability proportional
to the volume of the simplex they define. In volume sampling, one
picks a subset of columns $S$ of size $k$ with probability

$$
P_S = \frac{Vol(S)^2}{\sum_{T:|T|=k} Vol(T)^2},
$$

\noindent where the summation in the denominator is over all
subsets of size $k$. This sampling provides an almost tight
low-rank approximation of a matrix in Frobenius norm. Improving
this existence result, Deshpande and Vempala
\cite{Deshpande-Vempala} provided an adaptive randomized algorithm
for the low-rank approximation problem, which includes a
sub-procedure that repetitively chooses a small number of columns
by approximating volume sampling. This algorithm is essentially a
greedy algorithm and can be regarded as a randomized version of
the greedy algorithm we have analyzed for MAX-VOL
\cite{Volume-TCS}. They show that, if $\tilde{P_S}$ is the
probability that this algorithm chooses a subset of columns $S$ of
size $k$, then

\begin{equation}
\tilde{P_S} \leq k! \cdot P_S. \label{ps}
\end{equation}

\noindent Thus, not only is sampling larger volume columns good,
but approximately sampling columns with large volume can prove
useful for matrix approximation. A natural question is to ask what
happens when one finds a set of columns with the largest volume
(deterministic), which is our problem MAX-VOL. Note that, the last
expression (\ref{ps}) is reminiscent of the approximation ratio we
have proved for MAX-VOL in \cite{Volume-TCS}, but its analysis
relies on a linear algebraic identity whereas the result in
\cite{Volume-TCS} is derived via combinatorial means. MAX-VOL and
volume sampling seem to be related, but they have different
characteristics. MAX-VOL is proven to be intractable by using
complexity theoretic tools, whereas according to a recent result
by Deshpande and Rademacher \cite{Desh-Rade}, volume sampling can
be exactly implemented in polynomial time. This work together with
\cite{Desh-Rade} reveals the fact that, although one can exactly
sample the columns of a matrix with probability proportional to
their volumes, identifying a subset with the maximum volume is
hard.

Goreinov and Tyrtyshnikov \cite{Tyr} provided explicit statements
of how MAX-VOL, in particular, is related to low-rank
approximations in the following theorem:

\begin{thm} {\rm \cite{Tyr}}
\label{kk} Suppose that $A$ is an $m \times n$ block matrix of the
form

$$
A = \left( \begin{array}{cc}
A_{11} & A_{12} \\
A_{21} & A_{22}\\
\end{array}
\right)
$$

\noindent where $A_{11}$ is nonsingular, $k \times k$, whose
volume is at least $\mu^{-1}$ times the maximum volume among all
$k \times k$ sub-matrices. Then \footnote{$\|B\|_{\infty}$ denotes
the maximum modulus of the entries of a matrix $B$.} $\|A_{22} -
A_{21} A_{11}^{-1} A_{12}\|_{\infty} \leq \mu (k+1)
\sigma_{k+1}(A)$.
\end{thm}

This theorem implies that if one has a good approximation to the
maximum volume $k \times k$ sub-matrix, then the rows and columns
corresponding to this sub-matrix can be used to obtain a good
approximation to the entire matrix in the $\infty$-norm. If
$\sigma_{k+1}(A)$ is small for some small $k$, then this yields a
low-rank approximation to $A$. \cite{Tyr2} also proves a similar
result to Theorem \ref{kk}.

Pan \cite{Pan} unifies the main approaches developed for finding
RRQR factorizations by defining the concept of \emph{local maximum
volume} and then gives a theorem relating it to the quality of
approximation.

\begin{defn} {\rm \cite{Pan}}
Let $A \in \mathbb{R}^{m \times n}$ and $C$ be a sub-matrix of $A$
formed by any $k$ columns of $A$. $Vol(C)(\neq 0)$ is said to be
local $\mu$-maximum volume in $A$, if $\mu \, Vol(C) \geq Vol(C')$
for any $C'$ that is obtained by replacing one column of $C$ by a
column of $A$ which is not in $C$.
\end{defn}

\begin{thm} {\rm \cite{Pan}}
\label{local} For a matrix $A \in \mathbb{R}^{n \times n}$, an
integer $k$ $(1 \leq k < n)$ and $\mu \geq 1$, let $\Pi \in
\mathbb{R}^{n \times n}$ be a permutation matrix such that the
first $k$ columns of $A\Pi$ is a local $\mu$-maximum in $A$. Then,
for the QR factorization

$$
A \Pi = Q \left( \begin{array}{cc}
R_{11} & R_{12} \\
0 & R_{22} \\
\end{array}
\right),
$$

\noindent we have $\sigma_{min}(R_{11}) \geq (1/\sqrt{k(n-k) \,
\mu^2+1}) \sigma_k(A)$ and $\sigma_1(R_{22}) \leq \sqrt{k(n-k) \,
\mu^2+1} \, \sigma_{k+1}(A)$.
\end{thm}

We note that, MAX-VOL asks for a stronger property of the set of
vectors to be chosen, i.e. it asks for a ``good'' set of vectors
in a global sense rather than only requiring local optimality.
Obviously, a solution to MAX-VOL provides a set of vectors with
local maximum volume.

Independently of our work, there are some results in computational
geometry which are related to the ability to construct large
simplices embedded in V-polytopes. Essentially, the problem we
consider is a more general version of finding a large simplex in a
V-polytope, where the vertices of the polytope are the column
vectors. The results in this area are similar to ours in spirit
but using different techniques \cite{Gritzmann,Koutis,Packer}. The
most relevant work to ours is that of Koutis \cite{Koutis}, which
shows exponential inapproximability for finding a large simplex in
a V-polytope. He provides a reduction from set packing using an
inapproximability result of \cite{Hazan}, whereas our reduction is
directly from the Label Cover problem.

%%%%%%%%%%%%%%%%%%%%%%%%%%%%%%%%%%%%%%%%%%%%%%%%%%%%%%%%%%%%%%%%%%%%%%%%%%%%%

\section{The Label-Cover Problem}
Our reduction will be from the Label Cover problem. Label Cover
combinatorially captures the expressive power of a 2-prover
1-round proof system for the problem Max-3SAT(5). Specifically,
there exists a reduction from Max-3SAT(5) to Label Cover, so that
using the well known parallel repetition technique for the
specified proof system yields a new $k$-fold Label Cover instance.
For simplicity, we prefer to state our reduction from Label Cover
and for the sake of completeness, we provide a canonical reduction
from Max-3SAT(5) to Label Cover.

Max-3SAT(5) is defined as follows: Given a set of $5n/3$ variables
and $n$ clauses in conjunctive normal form where each clause
contains three distinct variables and each variable appears in
exactly five clauses, find an assignment of variables such that it
maximizes the fraction of satisfied clauses. The following result
is well known \cite{Pcp,Pcp2}:

\begin{thm}
\label{pcp-thm} There is a constant $\epsilon > 0$, such that it
is NP-hard to distinguish between the instances of Max-3SAT(5)
having optimal value $1$ and optimal value at most $(1-\epsilon)$.
\end{thm}

Although this result was proved for general 3CNF formulas, without
the requirement that each variable appears exactly $5$ times,
there is a standard reduction from Max-3SAT to Max-3SAT(5)
\cite{Feige}, which only results in a difference in the constant
$\epsilon$.

A Label Cover instance $L$ is defined as follows:
\begin{displaymath}L = (G(V,W,E),(\Sigma_V, \Sigma_W),\Pi)\end{displaymath} where

\begin{itemize}
\item $G(V,W,E)$ is a regular bipartite graph with vertex sets $V$
and $W$, and the edge set $E$. \vspace{1mm}

\item $\Sigma_V$ and $\Sigma_W$ are the label sets associated with
$V$ and $W$, respectively. \vspace{1mm}

\item $\Pi$ is the collection of constraints on the edge set,
where the constraint on an edge $e$ is defined as a function
$\Pi_e: \Sigma_V \rightarrow \Sigma_W$.
\end{itemize}

\noindent A labeling is an assignment to the vertices of the
graph, $\sigma: \{V \rightarrow \Sigma_V \} \cup \{W \rightarrow
\Sigma_W \}$. It is said to satisfy an edge $e=(v,w)$ if
$\Pi_e(\sigma(v)) = \sigma(w)$. The Label Cover problem asks for
an assignment $\sigma$ such that the fraction of the satisfied
edges is maximum.

A standard reduction from Max-3SAT(5) to Label Cover reveals that

\begin{thm}
There is a constant $\epsilon' > 0$, such that it is NP-hard to
distinguish between the instances of Label Cover having optimal
value $1$ and optimal value at most $(1-\epsilon')$.
\end{thm}

In order to amplify the gap, one can define a new Label Cover
instance for which the vertex set is essentially a set Cartesian
product of the original one. This instance, as follows, captures a
standard 2-prover 1-round protocol with parallel repetition $\ell$
times applied. We first note that for a given set $S = \{s_1,
\hdots ,s_n\}$, $S^{\ell}$ consists of all $\ell$-tuples of the
form $(s_{i_1}, \hdots , s_{i_{\ell}})$ where $s_{i_j} \in S$ and
$i_j$ runs over $\{1, \hdots ,n\}$ for $\ell \geq j \geq 1$. Given
the original Label Cover instance $L = (G(V,W,E),(\Sigma_V,
\Sigma_W),\Pi)$ reduced from Max-3SAT(5), let

$$L^{\ell} =
(G^{\ell}(V^{\ell},W^{\ell},E^{\ell}),\Sigma_V^{\ell},\Sigma_W^{\ell},\Pi^{\ell}),$$

\noindent where $V^{\ell}$, $W^{\ell}$, $\Sigma_V^{\ell}$ and
$\Sigma_W^{\ell}$ are the $\ell$ times Cartesian products of the
sets $V$, $W$, $\Sigma_V$ and $\Sigma_W$, respectively as defined
above. Let

\begin{itemize}
\item $E^{\ell}$ consist of all edges of the form $e = (v,w)$
where $v = (v_{i_1}, \hdots ,v_{i_{\ell}})$ and $w = (w_{i_1},
\hdots , w_{i_{\ell}})$ satisfying $(v_{i_j},w_{i_j}) \in E$ and
for all $\ell \geq j \geq 1$. \vspace{1.5mm}

\item $\Pi^{\ell}$ be the collection of constraints on the edge
set $E^{\ell}$. The constraint on an edge $e=(v,w)$ where $v =
(v_{i_1}, \hdots, v_{i_{\ell}})$ and $w = (w_{i_1}, \hdots,
w_{i_{\ell}})$ is a function $\Pi^{\ell}_e: \Sigma_V^{\ell}
\rightarrow \Sigma_W^{\ell}$ which is essentially an $\ell$-tuple
constraint ($\Pi^{\ell}_{e_1}, \hdots \Pi^{\ell}_{e_{\ell}}$),
where $\Pi^{\ell}_{e_j} = \Pi_{(v_{i_j},w_{i_j})}$ for $\ell \geq
j \geq 1$.

\end{itemize}

\noindent A labeling $\sigma$ of the vertices $V^{\ell}$ and
$W^{\ell}$ is said to satisfy an edge $e=(v,w)$ where $v =
(v_{i_1}, \hdots, v_{i_{\ell}})$ and $w = (w_{i_1}, \hdots,
w_{i_{\ell}})$, if $\Pi^{\ell}_e(\sigma(v)) = \sigma(w)$. Note
that this requirement is equal to
$\Pi_{(v_{i_j},w_{i_j})}(\sigma(v_{i_j})) = \sigma(w_{i_j})$ for
all $\ell \geq j \geq 1$. It is easy to see that, in this new
Label Cover instance, $|V| = (5n/3)^{\ell}$, $|W| = n^{\ell}$,
$|E| = (5n)^{\ell}$, $|\Sigma_V^{\ell}| = 7^{\ell}$ and
$|\Sigma_W^{\ell}| = 2^{\ell}$; the degrees of the vertices in $V$
and $W$ is $3^{\ell}$ and $5^{\ell}$, respectively. The following
theorem is a well known result by Raz \cite{Raz}:

\begin{thm}
\label{raz} There is an absolute constant $\alpha > 0$, such that
it is NP-hard to distinguish between the case that $OPT(L^{\ell})
= 1$ and $OPT(L^{\ell}) \leq 2^{-\alpha \ell}$.
\end{thm}

\section{Exponential Inapproximability of MAX-VOL}
\subsection{The Basic Gadget}
At the heart of our analysis is a set of vectors with a special
property. We will use a set of vectors (composed of binary entries
for simplicity of construction) such that any two of them have
large dot-product. We will also require that the dot product of a
vector and the binary complement of any other vector is large.
More specifically, we need these dot products be proportional to
the Euclidean norms squared of the vectors.

Given a vector $v = (v_1 \ldots v_m)$ where $v_i \in \{0,1\}$ for
$m \geq i \geq 1$, we denote the binary complement of $v$ by
$\overline{v} = (\overline{v_1} \ldots \overline{v_m})$ where
$\overline{v_i} = 1$ if $v_i = 0$, and $\overline{v_i} = 0$
otherwise. We begin with the following lemma:

\begin{lem}
\label{existence-vector} For $m \geq 2$, there exists a set of
vectors $B = \{b_1, \hdots, b_{2^m-1}\}$ of dimension $2^m$ with
binary entries such that the following three conditions hold:

\begin{enumerate}
\item ${\|b_i\|}_2 = 2^{(m-1)/2}$ for $2m-1 \geq i \geq 1$
\vspace{1.5mm} \item $b_i \cdot \overline{b_j} = 2^{m-2}$ for
$2m-1 \geq i
> j \geq 1$. \vspace{1.5mm} \item $b_i \cdot b_j = 2^{m-2}$ for $2m-1 \geq i > j
\geq 1$.
\end{enumerate}
\end{lem}

\begin{proof}
Consider the Hadamard matrix $H$ of dimension $2^m \times 2^m$
with entries $-1$ and $1$, constructed recursively by Sylvester's
method. Let $B$ be the $(2^m-1) \times 2^m$ matrix consisting of
the rows of $H$ for which we replace $-1$'s with $0$'s, excluding
the all $1$'s row. We claim that the rows of $B$ satisfy the
requirements. Indeed, by the properties of Hadamard matrices, each
row of $B$ has exactly $2^{m-1}$ $1$'s which satisfies the first
requirement. Note also that, for $m \geq 2$, two distinct rows of
$H$ (excluding the all $1$'s vector) have exactly $2^{m-2}$
element-wise dot-products of the following four types: $1 \cdot
1$, $1 \cdot (-1)$, $(-1) \cdot 1$, $(-1) \cdot (-1)$. Considering
the construction of $B$, we have that the dot-product of any two
of its rows is $2^{m-2}$ since all the products in $H$ involving
$-1$ vanishes for $B$. Similarly the dot-product of a row with the
binary complement of another row is $2^{m-2}$ by symmetry. Thus,
the second and the third requirement also hold.
\end{proof}

\subsection{The Reduction}
Lemma \ref{existence-vector} guarantees the existence of of a set
of binary vectors $B = \{b_1, \hdots, b_{2^{\ell}}\}$ of dimension
$2^{\ell+1}$ such that the following three conditions hold:

\begin{enumerate}
\item ${\|b_i\|}_2 = 2^{\ell/2}$ for $2^{\ell} \geq i \geq 1$
\vspace{1mm} \item $b_i \cdot \overline{b_j} = 2^{\ell-1}$ for
$2^{\ell} \geq i
> j \geq 1$. \vspace{1mm} \item $b_i \cdot b_j = 2^{\ell-1}$ for $2^{\ell} \geq i
> j \geq 1$.
\end{enumerate}

\noindent $B$ can be constructed in time $O(2^{2\ell})$. In our
reduction, $\ell$ will be a constant (to be exactly determined
later) inversely proportional to $\alpha$ which is the constant in
Raz' Theorem. Hence, one can construct $B$ in constant time. For
the sake of simplicity of our argument, we normalize the vectors
in $B$, which then clearly satisfies

\begin{enumerate}
\item ${\|b_i\|}_2 = 1$ for $2^{\ell} \geq i \geq 1$
\vspace{1.5mm} \item $b_i \cdot \overline{b_j} = 1/2$ for
$2^{\ell} \geq i > j \geq 1$. \vspace{1.5mm} \item $b_i \cdot b_j
= 1/2$ for $2^{\ell} \geq i > j \geq 1$.
\end{enumerate}

Given a Max-3SAT(5) instance and the reduction described in the
previous section, we will define a column vector for each
vertex-label pair in $L^{\ell}$, making $(35n/3)^{\ell} +
(2n)^{\ell}$ vectors in total. (Note that $|V^{\ell}| =
(5n/3)^{\ell}$, $|W^{\ell}| = n^{\ell}$,
$\Sigma_V^{\ell}=\{1,\hdots,7^{\ell}\}$ and $\Sigma_W^{\ell} =
\{1,\hdots 2^{\ell}\}$). Each vector will be composed of
$|E^{\ell}| = (5n)^{\ell}$ ``blocks'' which are either vectors
from the set $B$ or the zero vector according to the adjacency
information. More specifically, let $A_{v,i}$ be the vector for
the vertex label pair $v \in V^{\ell}$ and $i \in
\Sigma_V^{\ell}$. Similarly let $A_{w,j}$ be the vector for the
pair $w \in W$ and $j \in \Sigma_W^{\ell}$. Both of these vectors
are $(5n)^{\ell} 2^{\ell+1}$ dimensional. The block of $A_{v,i}$
corresponding to an edge $e \in E^{\ell}$ is denoted by
$A_{v,i}(e)$. The block of $A_{w,j}$ corresponding to an edge $e
\in E^{\ell}$ is denoted by $A_{w,j}(e)$. We define

$$
A_{v,i}(e) = \left\{
\begin{array}{ll}
\displaystyle\frac{\overline{b_{\Pi^{\ell}_e(i)}}}{3^{\ell/2}} &
\text{if } e \text{ is incident to $v$} \vspace{1.5mm} \\
\overrightarrow{0} & \text{if } e \text{ is not incident to } v .
\end{array}
\right.
$$

$$
A_{w,j}(e) = \left\{
\begin{array}{ll}
\displaystyle\frac{b_j}{5^{\ell/2}} & \text{if } e \text{ is
incident to } w \vspace{1.5mm}
\\ \overrightarrow{0} & \text{if } e \text{ is not incident to } w
\end{array} \right.
$$

\begin{figure}[b]
\centering

\setlength{\unitlength}{.4in}
\begin{picture}(10,5)(0,0)
\linethickness{1pt} \put(2,4){\circle*{0.15}}
\put(2,4){\line(1,0){6}} \put(2,2){\circle*{0.15}}
\put(2,2){\line(1,0){6}} \thicklines \put(2,4){\line(3,-1){6}}
\put(8,4){\circle*{0.15}} \put(8,2){\circle*{0.15}}
\put(1.5,4){\makebox(0,0){$v_1$}}
\put(1.5,2){\makebox(0,0){$v_2$}}
\put(8.5,4){\makebox(0,0){$w_1$}}
\put(8.5,2){\makebox(0,0){$w_2$}}
\put(4.8,4.2){\makebox(0,0){$e_1$}}
\put(4.8,1.8){\makebox(0,0){$e_3$}}
\put(5.6,3){\makebox(0,0){$e_2$}} \put(5,1){\circle*{0.08}}
\put(5,0.8){\circle*{0.08}} \put(5,0.6){\circle*{0.08}}
\end{picture}
\caption{A part of a simple bipartite graph representing a
Label-Cover instance} \label{fig:graph}
\end{figure}
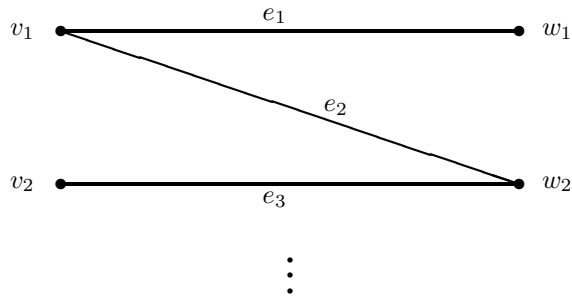

\begin{figure}[t]
\setlength{\unitlength}{.25in}
\begin{picture}(20,15)(0,0)
\linethickness{0.5pt} \put(1.5,13){\line(1,0){12.5}}
\put(1.5,12){\line(1,0){12.5}} \put(1.5,11){\line(1,0){12.5}}
\put(1.5,9){\line(1,0){12.5}} \put(1.5,8){\line(1,0){12.5}}
\put(1.5,6){\line(1,0){12.5}} \put(1.5,5){\line(1,0){12.5}}
\put(1.5,3){\line(1,0){12.5}} \put(1.5,2){\line(1,0){12.5}}

\put(1.5,13){\line(0,-1){2}} \put(1.5,9){\line(0,-1){1}}
\put(1.5,6){\line(0,-1){1}} \put(1.5,3){\line(0,-1){1}}

\put(4,13.5){\line(0,-1){12}} \put(6.5,13.5){\line(0,-1){12}}
\put(9,13.5){\line(0,-1){12}}

\put(14,13){\line(0,-1){2}} \put(14,9){\line(0,-1){1}}
\put(14,6){\line(0,-1){1}} \put(14,3){\line(0,-1){1}}

\put(0.8,12.5){\makebox(0,0){$A_{v_1,1}$}}
\put(0.8,11.5){\makebox(0,0){$A_{v_1,2}$}}
\put(0.8,8.5){\makebox(0,0){$A_{v_2,1}$}}
\put(0.8,5.5){\makebox(0,0){$A_{w_1,1}$}}
\put(0.8,2.5){\makebox(0,0){$A_{w_2,1}$}}

\put(2.7,13.5){\makebox(0,0){$e_1$}}
\put(5.2,13.5){\makebox(0,0){$e_2$}}
\put(7.7,13.5){\makebox(0,0){$e_3$}}

\put(2.7,12.5){\makebox(0,0){$\overline{a_{e_1}(1)}$}}
\put(5.2,12.5){\makebox(0,0){$\overline{a_{e_2}(1)}$}}
\put(7.7,12.5){\makebox(0,0){$\overrightarrow{0}$}}
\put(11.5,12.5){\makebox(0,0){$\overrightarrow{0}$}}

\put(2.7,11.5){\makebox(0,0){$\overline{a_{e_1}(2)}$}}
\put(5.2,11.5){\makebox(0,0){$\overline{a_{e_2}(2)}$}}
\put(7.7,11.5){\makebox(0,0){$\overrightarrow{0}$}}
\put(11.5,11.5){\makebox(0,0){$\overrightarrow{0}$}}

\put(2.7,8.5){\makebox(0,0){$\overrightarrow{0}$}}
\put(5.2,8.5){\makebox(0,0){$\overrightarrow{0}$}}
\put(7.7,8.5){\makebox(0,0){$\overline{a_{e_3}(1)}$}}
\put(11.5,8.5){\makebox(0,0){$\overrightarrow{0}$}}

\put(2.7,5.5){\makebox(0,0){$a(1)$}}
\put(5.2,5.5){\makebox(0,0){$\overrightarrow{0}$}}
\put(7.7,5.5){\makebox(0,0){$\overrightarrow{0}$}}
\put(11.5,5.5){\makebox(0,0){$\overrightarrow{0}$}}

\put(2.7,2.5){\makebox(0,0){$\overrightarrow{0}$}}
\put(5.2,2.5){\makebox(0,0){$a(1)$}}
\put(7.7,2.5){\makebox(0,0){$a(1)$}}
\put(11.5,2.5){\makebox(0,0){$\overrightarrow{0}$}}

\put(2,9.8){\circle*{0.07}} \put(2,10){\circle*{0.07}}
\put(2,10.2){\circle*{0.07}}

\put(2,6.8){\circle*{0.07}} \put(2,7){\circle*{0.07}}
\put(2,7.2){\circle*{0.07}}

\put(2,3.8){\circle*{0.07}} \put(2,4){\circle*{0.07}}
\put(2,4.2){\circle*{0.07}}

\put(11,13.3){\circle*{0.07}} \put(11.3,13.3){\circle*{0.07}}
\put(11.6,13.3){\circle*{0.07}}

\put(11,9.3){\circle*{0.07}} \put(11.3,9.3){\circle*{0.07}}
\put(11.6,9.3){\circle*{0.07}}

\put(11,6.3){\circle*{0.07}} \put(11.3,6.3){\circle*{0.07}}
\put(11.6,6.3){\circle*{0.07}}

\put(11,3.3){\circle*{0.07}} \put(11.3,3.3){\circle*{0.07}}
\put(11.6,3.3){\circle*{0.07}}

\put(16.5,8.5){\makebox(0,0){$\overline{a_e(i)} = \displaystyle
\frac{\overline{b_{\Pi_e^{\ell}(i)}}}{3^{\ell/2}}$}}
\put(16.5,7){\makebox(0,0){$a(j) = \displaystyle
\frac{b_j}{5^{\ell/2}}$}}
\end{picture}

\caption{The resulting (row) vectors in MAX-VOL instance computed
from the graph in Figure \ref{fig:graph} by our reduction}
\label{fig:matrix}
\end{figure}

\noindent In order to show how our reduction works, we present a
part of a simple bipartite graph in Figure ~\ref{fig:graph} with
all the edges drawn between two pairs of nodes, and the
corresponding (row) vectors computed by the reduction in Figure
~\ref{fig:matrix}. Note that $A_{v,i}$ has exactly $3^{\ell}$
non-zero blocks, and $A_{w,j}$ has $5^{\ell}$ non-zero blocks.
Hence, according to the definition above, their Euclidean norm is
$1$. The column vector set for the MAX-VOL instance is defined as

$$
A \in \mathbb{R}^{M \times N} = \{A_{v,i}| v \in V^{\ell}, i \in
\Sigma_V^{\ell}\} \cup \{A_{w,j}| w \in W^{\ell}, j \in
\Sigma_W^{\ell}\}.
$$

\noindent Note that $M = (5n)^{\ell} 2^{\ell+1}$ and $N =
(35n/3)^{\ell} + (2n)^{\ell}$, both having polynomial size in $n$
for constant $\ell$. From an intuitive point of view, we define
mutually orthogonal subspaces for each edge, and then we
``spread'' the Euclidean norm of each vector to the subspaces
corresponding to the edges incident to the vertex corresponding to
the vector. A crucial observation for this construction is that,
vectors $A_{v_1,i_1}$ and $A_{v_2,i_2}$ are orthogonal to each
other for all $v_1, v_2 \in V^{\ell}$, and $i_1, i_2 \in
\Sigma_V^{\ell}$, since there are no edges between the vertices in
$V^{\ell}$. The same result holds for the vertices in $W^{\ell}$.
From now on, this fact will be used frequently without explicit
reference. We set the number of column vectors $k$ to be chosen in
the MAX-VOL instance to $|V^{\ell}|+|W^{\ell}| = (5n/3)^{\ell} +
n^{\ell}$. Note that $k$ is a constant fraction of $N$, the total
number of columns, i.e. there exists a constant $\delta < 1$ such
that $k = \delta N$.

\subsection{Analysis}
We start with the completeness of the reduction:

\begin{thm}
If the Label Cover instance $L^{\ell}$ has a labeling that
satisfies all the edges, then in the MAX-VOL instance, there exist
$k$ column vectors with volume $1$.
\end{thm}

\begin{proof}
We show that there are at least $k$ orthogonal vectors. For an
edge $e=(v,w)$, let $i \in \Sigma_V^{\ell}$ and $j \in
\Sigma_W^{\ell}$ be the labeling of $v$ and $w$ assigned by the
optimal labeling which satisfies all the edges. Then, in the
MAX-VOL instance the dot product of the vectors $A_{v,i}$ and
$A_{w,j}$ is

\begin{equation}
\label{completeness} A_{v,i} \cdot A_{w,j} = \sum_{e \in E^{\ell}}
A_{v,i}(e) \cdot A_{w,j}(e) = \overline{b_{\Pi^{\ell}_e(i)}} \cdot
b_j = \overline{b_j} \cdot b_j = 0.
\end{equation}

\noindent This is due to the fact that the labeling satisfies $e$,
i.e. $b_{\Pi^{\ell}_e(i)} = b_j$. Since all the edges are
satisfied, and there exists a vector from each vertex
corresponding to the optimal labeling satisfying the equation
(\ref{completeness}), we have $|V^{\ell}|+|W^{\ell}|$ orthogonal
vectors, i.e. we have $k$ orthogonal vectors.
\end{proof}

Before proving the soundness of the reduction, which will prove
hardness of approximation, we first give the intuition for the
argument. According to our construction of the MAX-VOL instance,
there is a set of vectors corresponding to each node in $V^{\ell}$
and $W^{\ell}$. The set of vectors defined for a specific node has
high pair-wise dot products whereas a vector from a node $v_1 \in
V^{\ell}$ and another from $v_2$ in $V^{\ell}$ are orthogonal to
each other. The same goes for the vectors defined for $W^{\ell}$.
Hence, if vectors are chosen from the same set corresponding to a
single node, the total volume will decrease exponentially with
respect to the number of such vectors. Let us call these vectors
\emph{duplicates} in $V^{\ell}$ and $W^{\ell}$. The more intricate
part of the analysis is due to the dot products \emph{between} the
vectors defined for $V^{\ell}$ and $W^{\ell}$, which is enforced
to be non-zero by the unsatisfied edges in the Label-Cover
instance. We will show that, in case the Label-Cover instance has
few satisfied edges, any $k$ vectors chosen in the MAX-VOL
instance should satisfy the following: either the number of
duplicates in $V^{\ell}$ and $W^{\ell}$ is large enough so that
the total volume is small, or the dot products between $V^{\ell}$
and $W^{\ell}$ leads to a small volume.

\begin{thm}
\label{inapprox-thm} There exist absolute constants $\alpha$ and
$c$ such that, if the Label Cover instance $L^{\ell}$ does not
have any labeling that satisfies more than $2^{-\alpha \ell}$ of
the edges, then the volume of any $k$ vectors in the MAX-VOL
instance is at most $2^{-ck}$.
\end{thm}

\begin{proof}
Let $V^{\ell} = \{v_1, \hdots ,v_{(5n/3)^{\ell}}\}$ and $W^{\ell}
= \{w_1, \hdots ,w_{n^{\ell}}\}$. Let $A_v$ be the vectors
corresponding to the vertex $v \in V^{\ell}$: $A_v = \{A_{v,i} | i
\in \Sigma_V^{\ell}\}$. Similarly, let $A_{w} = \{A_{w,j} | j \in
\Sigma_W^{\ell}\}$ for $w \in W^{\ell}$. Let $A_{V^{\ell}}$ be the
set of all vectors corresponding to the nodes in $V^{\ell}$, and
$A_{W^{\ell}}$ be the set of all vectors corresponding to the
nodes in $W^{\ell}$, i.e.

$$A_{V^{\ell}} = \bigcup_{i=1}^{(5n/3)^{\ell}} A_{v_i}, \qquad A_{W^{\ell}} =
\bigcup_{i=1}^{n^{\ell}} A_{w_i}.
$$

\noindent \noindent For a set of vectors $C$ of size $k$, let $C_u
= C \cap A_u$ for all $u \in \{V^{\ell} \cup W^{\ell}\}$,
$C_{V^{\ell}} = C \cap A_{V^{\ell}}$ and $C_{W^{\ell}} = C \cap
A_{W^{\ell}}$. Let $V^{\ell}(C)$ and $W^{\ell}(C)$ be the set of
vectors for which $C$ ``selects'' at least one vector from
$V^{\ell}$ and $W^{\ell}$, respectively.

$$
V^{\ell}(C) = \{v \in V^{\ell}| C_v \neq \emptyset\}, \qquad
W^{\ell}(C) = \{w \in W^{\ell}| C(A_w) \neq \emptyset\}.
$$

\noindent For ease of notation, we let $k_{V_C} = |C_{V^{\ell}}|,
k_{W_C} = |C_{W^{\ell}}|, d_{V_C} = k_{V_C}-|V^{\ell}(C)|, d_{W_C}
= k_{W_C}-|W^{\ell}(C)|$. Note that $k_{V_C}$ and $k_{W_C}$ denote
how many vectors are chosen by $C$ from $V^{\ell}$ and $W^{\ell}$,
respectively. Whereas $d_{V_C}$ and $d_{W_C}$ are the total number
of duplicates in $C_{V^{\ell}}$ and $C_{W^{\ell}}$, respectively.
The following lemma relates the number of duplicates on one side
with its volume.

\begin{lem}
\label{exp_drop} $Vol(C_{V^{\ell}}) \leq (\sqrt{3}/2)^{d_{V_C}}$
and $Vol(C_{W^{\ell}}) \leq (\sqrt{3}/2)^{d_{W_C}}$.
\end{lem}

\begin{proof}
Let $P$ be the set of $|V^{\ell}(C)|$ elements which contains
exactly one vector of the form $A_{v,i}$ for each $v \in
V^{\ell}(C)$. In words, we consider the vectors of $C$
corresponding to the nodes in the Label-Cover instance minus all
the duplicates. For the duplicate vector $A_{v,j}$, we have
$A_{v,i} \cdot A_{v,j} = 1/2$. Hence, $d(A_{v,j}, P) \leq
d(A_{v,j}, A_{v,i}) = \sqrt{3}/2$. By the definition of $d_{V_C}$
and by the Union Lemma, we get $Vol(C_{V^{\ell}}) \leq
(\sqrt{3}/2)^{d_{V_C}}$. The argument for $Vol(C_{W^{\ell}})$ is
similar.
\end{proof}

Let the constant $c = 1/(3\cdot5^{\ell+1})$. Recall that, our
reduction will require $\ell$ to be inversely proportional to
$\alpha$ in Raz' Theorem. Hence, although having an exponential
dependence on $\alpha$, $c$ is a constant. We will show that
Theorem \ref{inapprox-thm} holds for this value of $c$; we will
prove that $Vol(C) \leq 2^{-ck}$ for any set $C$ of $k$ vectors.
To this aim, we argue by contradiction. The next lemma roughly
states that if the volume of $C$ is large enough, then its vectors
are almost equally distributed among the nodes of the Label-Cover
instance. This condition will in turn imply a small volume
completing our argument.

\begin{clm}
\label{homo} If $Vol(C) \geq 2^{-ck}$ for $c =
1/(3\cdot5^{\ell+1})$, then

\begin{equation}
\label{eq:k_V} (1-\epsilon_1) (5n/3)^{\ell} < k_{V_C} <
(1+\epsilon_1) (5n/3)^{\ell},
\end{equation}

\begin{equation}
\label{eq:k_W} (1-\epsilon_2) n^{\ell} < k_{W_C} < (1+\epsilon_2)
n^{\ell},
\end{equation}

\noindent where $\epsilon_1 =
\frac{1}{3^{\ell+1}}\left((3/5)^{\ell} + (3/5)^{2\ell}\right)$ and
$\epsilon_2 = \frac{1}{3^{\ell+1}} \left((3/5)^{\ell} + 1\right)$.
\end{clm}

\begin{proof}
First, we note that $Vol(C) \leq Vol(C_{V^{\ell}})$ since all the
vectors in the MAX-VOL instance have unit norm. Similarly, $Vol(C)
\leq Vol(C_{W^{\ell}})$. Thus, by the premise of the claim, we
have $Vol(C_{V^{\ell}}) \geq 2^{-ck}$ and $Vol(C_{W^{\ell}}) \geq
2^{-ck}$. By Lemma \ref{exp_drop}, we get

$$
(\sqrt{3}/2)^{d_{V_C}} = 2^{d_{V_C}(-1+\log{3}/2)} \geq
Vol(C_{V^{\ell}}) \geq 2^{-ck}
$$

\noindent which implies $d_{V_C} \leq ck/(1-\log{3}/2) < 5ck$
since $\log{3} < 1.6$. The analysis for $d_{W_C}$ along exactly
the same lines also yields $d_{W_C} < 5ck$. Noting the expressions
for $c$ and $k$, and following the definitions, we obtain

\begin{align*}
k_{V_C} = |V^{\ell}(C)| + d_{V_C} &< |V^{\ell}| + 5ck \\
&= (5n/3)^{\ell} + \frac{1}{3 \cdot 5^{\ell}} ((5n/3)^{\ell}+n^{\ell}) \\
&= (1+\epsilon_1) (5n/3)^{\ell}.
\end{align*}

\noindent Similarly,

\begin{align*}
k_{W_C} = |W^{\ell}(C)| + d_{W_C} &< |W^{\ell}| + 5ck \\
&= n^{\ell} + \frac{1}{3 \cdot 5^{\ell}} ((5n/3)^{\ell}+n^{\ell}) \\
&= (1+\epsilon_2) n^{\ell}
\end{align*}

\noindent which proves the right hand sides of ~\eqref{eq:k_V} and
~\eqref{eq:k_W}. Noting that $k_{V_C} + k_{W_C} = k =
(5n/3)^{\ell} + n^{\ell}$, we get

\begin{align*}
k_{V_C} = k-k_{W_C} &> (5n/3)^{\ell} + n^{\ell} - (1+\epsilon_2) n^{\ell} \\
&= (5n/3)^{\ell} - \frac{1}{3^{\ell+1}}((3n/5)^{\ell} + n^{\ell}) \\
&= (1-\epsilon_1) (5n/3)^{\ell}
\end{align*}

\noindent and

\begin{align*}
k_{W_C} = k-k_{V_C} &> (5n/3)^{\ell} + n^{\ell} - (1+\epsilon_1) (5n/3)^{\ell} \\
&= n^{\ell} - \frac{1}{3^{\ell+1}}((3n/5)^{\ell} + n^{\ell}) \\
&= (1-\epsilon_2) n^{\ell}
\end{align*}

\noindent which proves the left hand sides.
\end{proof}

\noindent Claim \ref{homo} ensures that if the volume of a set of
$k$ vectors exceeds $2^{-ck}$, then some certain concentration
result should hold, namely Equation (\ref{eq:k_V}) and Equation
(\ref{eq:k_W}). We will now show that, these equations imply
$Vol(C) < 2^{-ck}$ which is our contradiction.

Without loss of generality, let $V^{\ell}(C) = \{v_1, \hdots,
v_q\}$, $W^{\ell}(C) = \{w_1, \hdots ,w_p\}$. Note that these sets
contain the nodes of the Label-Cover instance from which $C$
``selects'' at least one vector. Let $Q = \{A_{v_1,i_1}, \hdots ,
A_{v_q,i_q}\}$ where $A_{v_s,i_s} \in C_{v_s}$ for $s = 1, \hdots
,q$. Let $P = \{A_{w_1,j_1}, \hdots ,A_{w_p,j_p}\}$ where
$A_{v_s,i_s} \in C_{v_s}$ for $s = 1, \hdots ,p$. By definition,

$$
q = k_{V_C}-d_{V_C} > (1-2\epsilon_1)(5n/3)^{\ell}, \qquad p =
k_{W_C}-d_{W_C} > (1-2\epsilon_2)n^{\ell}.
$$

\noindent In words, the set of nodes from which $C$ selects at
least one vector essentially covers $V^{\ell}$ and $W^{\ell}$.
These vectors are all orthogonal. From this point of view,
$V^{\ell}(C))$ and $W^{\ell}(C))$ play an important role in our
argument. Since $C$ ``covers'' $V^{\ell}$ and $W^{\ell}$ and since
the Label-Cover instance has many unsatisfied edges, it means that
the dot products of many vectors in $C_{V^{\ell}}$ with many
vectors in $C_{W^{\ell}}$ will be large. This will lead to small
volume. Hence, we are essentially interested in the number of
unsatisfied edges between $V^{\ell}(C)$ and $W^{\ell}(C)$. Since
there are at most $2^{-\alpha \ell}$ satisfied edges in the
Label-Cover instance, and there are exactly $3^{\ell}$ edges
incident to a node in $V^{\ell}$, the number of unsatisfied edges
incident to $V^{\ell}(C)$ is greater than
$(1-2\epsilon_1-2^{-\alpha \ell})(5n)^{\ell}$. Similarly, the
number of unsatisfied edges incident to $W^{\ell}(C)$ is greater
than $(1-2\epsilon_2-2^{-\alpha \ell})(5n)^{\ell}$. Thus, the
number of unsatisfied edges whose end points are in $V^{\ell}(C)$
and $W^{\ell}(C)$, is greater than
$(1-2\epsilon_1-2\epsilon_2-2^{-\alpha \ell+1})(5n)^{\ell}$.

We now give an upper bound for the distance of the vectors in $Q$
to $P$, namely ${\|A_{v_s,i_s} - \pi_P(A_{v_s,i_s})\|}_2$ for each
$A_{v_s,i_s} \in Q$. To this end, we define the set
$N(A_{v_s,i_s}) = \{C_w|e=(A_{v_s,i_s},w) \text{ is
unsatisfied}\}$. Note that the vectors in different sets are
mutually orthogonal, and by the reduction we have

$$A_{v_s,i_s} \cdot A_{w,j} = \sum_{e \in E^{\ell}} A_{v_s,i_s}(e)
\cdot A_{w,j}(e) = \overline{b_{\Pi^{\ell}_e(i)}} \cdot b_j =
\frac{1}{2\cdot3^{\ell/2}\cdot5^{\ell/2}}$$

\noindent for $A_{w,j} \in N(A_{v_s,i_s})$ since
$e=(A_{v_s,i_s},A_{w,j})$ is unsatisfied. Thus, by the Pythagoras
Theorem, we obtain

$$
d(A_{v_s,i_s}, P) = {\|A_{v_s,i_s}-\pi_P(A_{v_s,i_s})\|}_2 <
\left(1-\frac{|N(A_{v_s,i_s})|}{4\cdot 3^{\ell} \cdot
5^{\ell}}\right)^{\frac{1}{2}}.
$$

\noindent Using the Union Lemma, we get

\begin{align*}
Vol(P \cup Q) &\leq Vol(P) \cdot \prod_{s=1}^q d(A_{v_s,i_s}, P) \\
&< Vol(P) \cdot \prod_{s=1}^q
\left(1-\frac{|N(A_{v_s,i_s})|}{4\cdot 3^{\ell} \cdot
5^{\ell}}\right)^{\frac{1}{2}}. \\
\end{align*}

\noindent The product in the last expression is maximized when all
the factors are equal to each other. We also previously showed
that $\sum_{s=1}^q |N(A_{v_s,i_s})| >
(1-2\epsilon_1-2\epsilon_2-2^{-\alpha \ell+1})(5n)^{\ell}$ and
that $q$, the number of distinct nodes hit in $V^{\ell}$
satisfies, $q > (1-2\epsilon_1)(5n/3)^{\ell}$. Hence, we obtain

\begin{align*}
Vol(P \cup Q) &< Vol(P) \cdot \prod_{s=1}^q
\left(1-\frac{\sum_{s=1}^q |N(A_{v_s,i_s})|}{q \cdot 4 \cdot 3^{\ell} \cdot 5^{\ell}}\right)^{\frac{1}{2}} \\
&< Vol(P) \cdot \prod_{s=1}^q
\left(1-\frac{(1-2\epsilon_1-2\epsilon_2-2^{-\alpha
\ell+1})(5n)^{\ell}}{(5n/3)^{\ell} \cdot 4\cdot 3^{\ell} \cdot 5^{\ell}}\right)^{\frac{1}{2}} \\
&= Vol(P) \cdot
\left(1-\frac{(1-2\epsilon_1-2\epsilon_2-2^{-\alpha \ell+1})}{4
\cdot 5^{\ell}}\right)^{\frac{q}{2}}
\\
&< Vol(P) \cdot
\left(1-\frac{(1-2\epsilon_1-2\epsilon_2-2^{-\alpha \ell+1})}{4
\cdot 5^{\ell}}\right)^{\frac{(1-2\epsilon_1)(5n/3)^{\ell}}{2}}.
\end{align*}

\noindent To simplify, let $t = \frac{4 \cdot
5^{\ell}}{(1-2\epsilon_1-2\epsilon_2-2^{-\alpha \ell+1})}$. For
$\ell \geq 1$, we have

$$\epsilon_1 = \frac{1}{3^{\ell+1}}\left((3/5)^{\ell} +
(3/5)^{2\ell}\right) \leq \frac{1}{3^2}\left((3/5) +
(3/5)^2\right) < 3/20.$$

\noindent Noting that $\log{e} \geq 10/7$, we obtain $\log{e}
\cdot (1-2\epsilon_1) \geq 10/7 \cdot 7/10 = 1$. Then, we get

\begin{align*}
Vol(P \cup Q) &< Vol(P) \cdot \left(1-\frac{1}{t}\right)^{t \cdot
\frac{(1-2\epsilon_1)(5n/3)^{\ell}}{2t}} \\
&\leq e^{-\frac{(1-2\epsilon_1)(5n/3)^{\ell}}{2t}} \\
&= 2^{-\log{e} \cdot \frac{(1-2\epsilon_1)(5n/3)^{\ell}}{2t}} \\
&\leq 2^{-\frac{(5n/3)^{\ell}}{2t}},
\end{align*}

\noindent where $e$ is the base of the natural logarithm. In the
second inequality, we have used the fact that $Vol(P) \leq 1$ and
$(1-\frac{1}{t})^t \leq e^{-1}$ for $t>1$.

We will now provide an upper bound for $t$ to further simplify the
last expression. To this aim, let $\ell'$ be the smallest integer
such that $2^{-\alpha \ell'+1} \leq 11/27$. Taking logarithms and
rearranging, it is easy to see that $\ell' = \left\lceil
\frac{\log{(\frac{54}{11})}}{\alpha} \right\rceil$. Note also that
for $\ell \geq 2$, we have $2\epsilon_1 < 1/27$ and $2\epsilon_2 <
3/27$. Then, for $\ell = \ell'$, we get

$$
t = \frac{4 \cdot 5^{\ell}}{(1-2\epsilon_1-2\epsilon_2-2^{-\alpha
\ell+1})} < \frac{4 \cdot
5^{\ell}}{(1-\frac{1}{27}-\frac{3}{27}-\frac{11}{27})} = \frac{4
\cdot 5^{\ell}}{(4/9)} = 9 \cdot 5^{\ell}.
$$

\noindent Since $k = (5n/3)^{\ell}+n^{\ell}$, we also have

$$
n^{\ell} = k/(1+(5/3)^{\ell}) > k/(5/3)^{\ell+1},
$$

\noindent which yields

$$
Vol(P \cup Q) < 2^{-\frac{(5n/3)^{\ell}}{9 \cdot 5^{\ell}}} =
2^{-\frac{n^{\ell}}{3^{\ell+2}}} < 2^{-\frac{k}{3 \cdot
5^{\ell+1}}} = 2^{-ck},
$$

\noindent which is our contradiction. Thus, the volume of a set of
$k$ vectors in a negative instance of MAX-VOL cannot exceed
$2^{-ck}$ for $c = \frac{1}{3 \cdot 5^{\ell+1}}$.
\end{proof}

We have shown that

\begin{itemize}
\item if the optimal value of the Label Cover instance is $1$,
then the optimal value of the MAX-VOL instance is $1$.
\vspace{1mm} \item if the optimal value of the $\ell$-fold Label
Cover instance is less than $2^{-\alpha \ell}$, then the optimal
value of the MAX-VOL instance is less than $2^{-ck}$.
\end{itemize}

\noindent By the combination of Theorem \ref{pcp-thm} and Theorem
\ref{raz}, we know that there exists a gap producing reduction
from SAT to $\ell$-fold Label Cover with parameters $1$ and
$2^{-\alpha \ell}$. This means that there is a polynomial time
reduction from SAT to MAX-VOL such that, given a formula $\phi$

\begin{itemize}
\item if $\phi$ is satisfiable , then $OPT(\text{MAX-VOL}) = 1$.
\vspace{1mm} \item if $\phi$ is not satisfiable , then
$OPT(\text{MAX-VOL}) < 2^{-ck}$.
\end{itemize}
%
%\noindent Now, assume that there exists an algorithm for MAX-VOL
%with approximation ratio greater than or equal to $2^{-ck}$. Then,
%
%\begin{itemize}
%\item if $\phi$ is satisfiable, the algorithm returns a subset
%with volume at least $2^{-ck}$. \item if $\phi$ is not
%satisfiable, then the volume of the subset returned by the
%algorithm is smaller than $2^{-ck}$.
%\end{itemize}
%
%\noindent which means that we can solve SAT in polynomial time.
Thus, unless $P=NP$, MAX-VOL is inapproximable within $2^{-ck}$
for some constant $c>0$.

\section{Discussion}
Our reduction heavily relies on the Raz' Parallel Repetition
Theorem \cite{Raz}. Indeed, it doesn't seem possible to get an
exponential inapproximability result without parallel repetition.
But, since the degrees of the vertices in the Label-Cover instance
exponentially increases with respect to the number of repetitions,
our constant $c$ depends on the constant $\alpha$ in Raz' result.
It might be possible to improve this constant by making use of
more sophisticated parallel repetition theorems, but we did not
proceed so far. Indeed, the exact analysis is irrelevant as the
constant will be too small in all cases. Overall, the strength of
our result is is directly related to the underlying theorems for
the inapproximability of  Label-Cover.

Another way of getting a stronger hardness result is to find a
more sophisticated reduction. In our MAX-VOL instance, the
subspaces ``reserved'' for each edge in the Label-Cover instance
are orthogonal to each other. This dramatically simplifies the
analysis, yielding perfect completeness, i.e. volume $1$ in
MAX-VOL. It might be possible to construct a MAX-VOL instance for
which these subspaces have some pair-wise angle, so that we
sacrifice the perfect completeness, but at the same time get a
much smaller soundness. This would improve the inapproximability
result.

The obvious open problem is whether the inapproximability can be
strengthened to $2^{-k+1}$. Recall that this is the lower bound
for the greedy algorithm for MAX-VOL. Considering the
multiplicative nature of the problem yielding a very small
approximation ratio for the obvious greedy algorithm, a
significant improvement of the upper bound would be expected to
provide asymptotically better approximations in the exponent. This
suggests that the inherent hardness of MAX-VOL might be very close
to the performance of the greedy algorithm. However, with the
techniques we have used, it is not possible to break the
dependence of $c$ on the constant in the parallel repetition
theorems.

We would finally like to point out that the reduction and the
analysis provided in this paper might be a good starting point for
studying hardness of other matrix approximation problems in
general (e.g. \cite{Boutsidis,Vempala}) for which no technique
related to the PCP theorem have been used. Such an extension to
other matrix approximation problems is not trivial. Indeed,
computing the volume is already difficult although purely
geometric intuition is used. Relating this to other linear
algebraic functions (e.g. singular values) which continuously
depend on the entries will put even more strain on the analysis.
\\

\textbf{Acknowledgments:} We would like to thank Ioannis Koutis
who, in the final stages of this paper, pointed out to us the
relevant lines of research in V-polytope theory \cite{Koutis}.

{
\bibliographystyle{abbrv}
\bibliography{inapprox-reference}
}

\end{document}